\documentclass[letterpaper, 10 pt, conference]{ieeeconf}  

\IEEEoverridecommandlockouts                              

\usepackage{bm}
\usepackage{float}
\usepackage[pdftex]{graphicx}
\usepackage{tabularx}
\usepackage{verbatim}
\usepackage{color}
\usepackage{xparse}
\usepackage{hyperref}
\usepackage{xmpmulti}
\usepackage{transparent}
\usepackage{fancyhdr}
\usepackage{cite}
\usepackage{wrapfig}
\usepackage{rotating}

\usepackage[ruled, lined, linesnumbered, commentsnumbered, longend]{algorithm2e}

\usepackage{booktabs}
\usepackage{caption}
\usepackage{subcaption}
\usepackage{tikz}
\usepackage{hyperref}
\usepackage{wrapfig}
\usepackage{setspace}
\usepackage{graphicx}
\usepackage[all]{xy}
\usepackage{verbatim}
\usepackage{float}
\usepackage{mathrsfs}  
\usepackage{bm}
\usepackage{multirow}
\usepackage{color,soul}
  
\usepackage{amssymb,amsmath,graphicx}
  
\usepackage{amsthm}
\usepackage[compact]{titlesec}
\usepackage{diagbox}
\usepackage{array}
\usepackage{titlesec}
\usepackage{units}
\usepackage[export]{adjustbox}
\usepackage{MnSymbol}
\usepackage{amssymb}

\let\oldIEEEkeywords\IEEEkeywords
\def\IEEEkeywords{\oldIEEEkeywords\normalfont\bfseries\ignorespaces}
\newtheorem{theorem}{Theorem} 
 
\newtheorem{lemma}{Lemma}            

\newtheorem{problem}{Problem}
                       
\newtheorem{definition}{Definition}

\newtheorem{corollary}{Corollary}
\newtheorem{remark}{Remark}

\newcommand{\argmin}{\mathrm{arg}\min}

\usepackage{mathtools}
\mathtoolsset{showonlyrefs}

\title{\LARGE \bf Distributed Differentially Private Control Synthesis for Multi-Agent Systems with Metric Temporal Logic Specifications}

\author{Nasim Baharisangari and Zhe Xu
\thanks{Nasim Baharisangari and Zhe Xu are with the School for Engineering of Matter, Transport and Energy, Arizona State University, Tempe, AZ 85287. {\tt\small $\{$nbaharis, xzhe1$\}$@asu.edu} (Corresponding author: Zhe Xu)}}

\titleformat{\section}{\centering\normalfont\scshape}{\thesection}{0em}{.~}
\titleformat{\paragraph}[runin]{\bfseries}{}{-2em}{}[:]

\DeclareMathAlphabet{\mathpzc}{OT1}{pzc}{m}{it}

\DeclareMathOperator{\luntil}{\mathbf{U}}
\DeclareMathOperator{\leventually}{\mathbf{F}}
\DeclareMathOperator{\lglobally}{\mathbf{G}}

\newcommand{\real}{\ensuremath{\mathbb{R}}}

\newcommand{\nat}{\ensuremath{\mathbb{N}}}

\newcommand{\distime}{\ensuremath{\mathbb{T}}}

\newcommand{\node}{\ensuremath{c}}

\newcommand{\Gennode}{\ensuremath{c}}

\newcommand{\nodeNo}{\ensuremath{n_C}}

\newcommand{\nodeSet}{\ensuremath{\mathcal{C}}}

\newcommand{\graphG}{\ensuremath{G}}

\NewDocumentCommand{\graphTraj}{O{}}{\ensuremath{g_{#1}}}

\NewDocumentCommand{\robustness}{O{}}{\ensuremath{r_{#1}}} 

\newcommand{\MinRob}{\ensuremath{r_{\text{\rm{min}}}}}

\newcommand{\counteri}{\ensuremath{i}}

\newcommand{\counterj}{\ensuremath{j}}

\newcommand{\dimcounter}{\ensuremath{d}}

\newcommand{\counterk}{\ensuremath{k}}

\newcommand{\timeIndex}{\ensuremath{t}}

\newcommand{\genPair}[2]{\ensuremath{(\node,\timeIndex)}}

\SetKwInput{Hyperparameter}{Hyperparameter}
\SetKwInput{Input}{Input}
\SetKwInput{Output}{Output}
\SetKwBlock{Loop}{loop}{end}

\SetKwInOut{Parameter}{Parameter}
\SetKwFunction{AlgoLearnLTL}{\text{Learn-LTL}}

\SetKwFunction{StopCriteria}{\text{Stopping-Criteria}}

\newcommand{\M}{\ensuremath{M}}

\newcommand{\horizon}{\ensuremath{H}}

\newcommand{\PiCon}[2]{\ensuremath{P_{#2}[#1]}}

\newcommand{\KalmanGain}[2]{\ensuremath{\boldsymbol{K}_{#2}[#1]}}

\newcommand{\StCovarianceMat}[2]{\ensuremath{\boldsymbol{\Sigma}}_{#2}[#1]}

\newcommand{\NoiseCo}[2]{\ensuremath{\boldsymbol{\mathpzc{K}}}^{#1}_{#2}}

\newcommand{\BoldControlInput}[2]{\ensuremath{\mathbf{U}_{#2}[#1]}}

\newcommand{\BoldControlInputHo}[3]{\ensuremath{\boldsymbol{u}^{#1}_{#2}{[#3]}}}

\newcommand{\BoldControlInputOld}[2]{\ensuremath{\mathbf{\Tilde{U}}^{#1}_{#2}}}

\newcommand{\ControlInput}[2]{\ensuremath{\boldsymbol{u}_{#2}[#1]}}

\newcommand{\NonBoldControlInput}[2]{\ensuremath{{u}_{#2}[#1]}}

\newcommand{\Env}{\ensuremath{\mathcal{S}}}

\newcommand{\EstSt}[2]{\ensuremath{\hat{\boldsymbol{s}}}_{#2}[#1]}

\newcommand{\EstStAgt}[2]{\ensuremath{\hat{{s}}}_{#2}[#1]}

\newcommand{\Boldy}[2]{\ensuremath{{\boldsymbol{y}}_{#2}[#1]}}

\newcommand{\y}[2]{\ensuremath{{{y}_{#2}[#1]}}}

\newcommand{\Noisyy}[2]{\ensuremath{{{\Tilde{y}}_{#2}[#1]}}}

\newcommand{\BoldNoisyy}[2]{\ensuremath{\boldsymbol{\Tilde{y}}}_{#2}[#1]}

\newcommand{\GM}[2]{\ensuremath{{\boldsymbol{\hat{\zeta}}}_{#2}{[#1]}}}

\newcommand{\GMNoT}[1]{\ensuremath{{\boldsymbol{\hat{\zeta}}}_{#1}}}

\newcommand{\GMnonBold}[2]{\ensuremath{{{\zeta}}_{#2}{[#1]}}}

\newcommand{\GMori}[2]{\ensuremath{{\boldsymbol{{\zeta}}}_{#2}{[#1]}}}

\newcommand{\GMoriNoT}[1]{\ensuremath{{\boldsymbol{{\zeta}}}_{#1}}}

\newcommand{\ActGM}[2]{\ensuremath{{{\eta}}_{#2}[#1]}}

\newcommand{\ActGMNoT}[1]{\ensuremath{{{\eta}}_{#1}}}

\newcommand{\AgtTrajNoT}[1]{\ensuremath{{{{{s}}}}_{#1}}}

\newcommand{\AgtTraj}[2]{\ensuremath{{{{{s}}}}_{#2}[#1]}}

\newcommand{\TopMatrix}{\ensuremath{{\boldsymbol{V}}}}

\newcommand{\ProbMatrix}[1]{\ensuremath{{\boldsymbol{W}}_{#1}}}

\newcommand{\GraphMat}{\ensuremath{{\boldsymbol{D}}}}

\newcommand{\EstErrorBound}[2]{\ensuremath{{\rho}_{#2}[#1]}}

\newcommand{\agtID}{\ensuremath{{{i}}}}

\newcommand{\St}[2]{\ensuremath{{{s_{#2}[#1]}}}}

\newcommand{\StNoT}[1]{\ensuremath{{{s_{#1}}}}}

\newcommand{\BoldSt}[2]{\ensuremath{{{\boldsymbol{s}_{#2}[#1]}}}}

\newcommand{\noise}[2]{\ensuremath{\boldsymbol{v}}^{#1}_{#2}}

\newcommand{\AugA}[2]{\ensuremath{{\boldsymbol{A}}}^{#1}_{#2}}

\newcommand{\AugB}[2]{\ensuremath{\boldsymbol{{B}}}^{#1}_{#2}}

\newcommand{\AugC}[2]{\ensuremath{\boldsymbol{{C}}^{#1}_{#2}}}

\newcommand{\Confidence}[2]{\ensuremath{\mathcal{P}}_{#2}[#1]}

\newcommand{\ConfidenceLe}[2]{\ensuremath{{\mathbb{P}}}^{#1}_{#2}}

\newcommand{\MinConf}{\ensuremath{\gamma_{\textrm{min}}}}

\newcommand{\STLtimeIndex}{\ensuremath{j}}

\newcommand{\formula}[2]{\ensuremath{\phi}^{#1}_{#2}}

\newcommand{\objective}{\ensuremath{J}}

\newcommand{\Horizon}{\ensuremath{H}}

\newcommand{\counterl}{\ensuremath{l}}

\newcommand{\MILP}{\ensuremath{\text{Diff-MILP}}}

\newcommand{\neighborSet}[1]{\ensuremath{\mathcal{N}}_{#1}}

\newcommand{\AgentSet}[1]{\ensuremath{\mathcal{Z}}_{#1}}

\newcommand{\DiffEps}[2]{\ensuremath{{\epsilon}^{#1}_{#2}}}

\newcommand{\DiffDelta}[2]{\ensuremath{{\delta}}^{#1}_{#2}}

\begin{document}

\maketitle
\thispagestyle{empty}
\pagestyle{empty}

\begin{abstract}\label{abstract}
In this paper, we propose a distributed differentially private receding horizon control (RHC) approach for multi-agent systems (MAS) with metric temporal logic (MTL) specifications. In the MAS considered in this paper, each agent privatizes its sensitive information from other agents using a differential privacy mechanism. In other words, each agent adds privacy noise (e.g., Gaussian noise) to its output to maintain its privacy and communicates its noisy output with its neighboring agents. We define two types of MTL specifications for the MAS: agent-level
specifications and system-level specifications. Agents should collaborate to satisfy the system-level MTL specifications with a minimum probability while each agent must satisfy its own agent-level MTL specifications at the same time. In the proposed distributed RHC approach, each agent communicates with its neighboring agents to acquire their noisy outputs and calculates an estimate of the system-level trajectory. Then each agent synthesizes its own control inputs such that the system-level specifications are satisfied with a minimum probability while the agent-level specifications are also satisfied. In the proposed optimization formulation of RHC, we directly incorporate Kalman filter equations to calculate the estimates of the system-level trajectory, and we use mixed-integer linear programming (MILP) to encode the MTL specifications as optimization constraints. Finally, we implement the proposed distributed RHC approach in a case study. 

\end{abstract}

\section{Introduction}\label{Sec:Introduction}

In multi-agent systems (MAS), it is common that agents collaborate to accomplish different types of system-level tasks through communication with each other, where the communication occurs among the agents that are neighbors \cite{KegeleirsSWAM}. Distributed control of an MAS, in comparison with centralized control has the advantages of scalability and fast computing \cite{GuanDist}\cite{ChenDist}\cite{DimitraDist}\cite{PelegDist}.
Furthermore, the centralized control can be computationally expensive, and if the central control unit fails, then the whole system may fail. In comparison,  distributed control has a better potential in fault tolerance \cite{PATTON2007280}. Distributed control has been used in many applications such as mobile robots \cite{MuDist} and autonomous underwater vehicles (AUVs) \cite{ZuoDist}. 

In an MAS, it is possible that while the agents are cooperating to satisfy system-level tasks through communication, each agent should protect its sensitive information (e.g., actual position state) from its neighboring agents \cite{XuDiff}. In such situations, \textit{differential privacy} can be employed to protect the privacy of the agents in an MAS. Differential privacy ensures that an adversary is not able to deduce an agent's sensitive information while allowing decision making on system level \cite{PeterDiff}\cite{DworkDiff}\cite{DworkDiff2}. For dynamical systems (e.g., multi-agent systems), differential privacy protects the privacy of each agent by adding \textit{differential privacy noise} (e.g., Gaussian noise) to the trajectories containing sensitive information such that an adversary is not able to deduce the privatized trajectories \cite{XuDiff}.

\textit{Metric temporal logic} (MTL) can be used to define different complicated tasks for an MAS due to being expressive and human-interpretable \cite{Linard} \cite{Seshia2016}. MTL is one type of temporal logics which is defined over real-valued data in discrete time domain \cite{Donze}. In addition, MTL is amenable to formal analysis, and these advantages make MTL a good candidate to define complicated tasks, such as collision avoidance \cite{Plaku2016MotionPW}, in the form of MTL formulas \cite{Haghighi2016RoboticSC}
\cite{FARAHANI2015323}\cite{xuuuuu}.


In the MAS considered in this paper, the agents collaborate with each other to satisfy system-level tasks while protecting their privacy and satisfying their own agent-level tasks at the same time, where the tasks are defined in the form of MTL formulas and each agent incorporates differential privacy to privatize its trajectory containing sensitive information. For satisfying the system-level tasks, each agent calculates an estimate of the system-level trajectory. First, each agent communicates with its neighboring agents to acquire their noisy outputs. Then, each agent employs Kalman filter to compute the estimates of the states of its neighboring agents and use the estimated information to estimate the system trajectory. In the next step, each agent uses receding horizon control (RHC) to synthesize control inputs for satisfying the system-level tasks and the agent-level tasks.

\paragraph{Contributions} We summarize our contributions as follows. (a) We propose a distributed differentially private receding horizon control (RHC) formulation for an MAS which is considered a stochastic dynamic system with MTL specifications. (b) In the proposed approach, we directly incorporate Kalman filter equations in the optimization formulation of RHC to account for the uncertainties stemming from differential privacy. In the proposed optimization formulation, we employ a \textit{one-step ahead predication} of the noisy outputs to be used in the Kalman filter equations. (c) In the proposed distributed RHC, by assigning individual tasks in addition to system-level tasks, we utilize a higher portion of the capacity of each agent in accomplishing different tasks. 


\section{Preliminaries}\label{Sec:Preliminaries}
In this section, we explain the notations, definitions, and concepts that we use in this paper. Table \ref{table:1} shows the important notations that we use in this section and the following sections.

\subsection{System Dynamics and Features of Multi-Agent Systems}\label{subsec: dynamics}

In this paper, an MAS consisting of $|\AgentSet{}|$ agents moves in a bounded environment $\Env\subseteq{\real}^{|\mathcal{D}|}$, where $\AgentSet{}$ denotes the set of the agents in the MAS, $|\AgentSet{}|$ denotes the cardinality of $\AgentSet{}$, and $\mathcal{D}=\{d_1,d_2,...,d_{|\mathcal{D}|}\}$ with $|\mathcal{D}|$ being the cardinality of the set $\mathcal{D}$. We represent the system dynamics of this MAS in the finite discrete time domain $\distime=\{1,2,...,\tau\}$ (where $\tau\in\nat=\{1,2,...\}$) with Eq. \eqref{eq:swarm}. 
\begin{equation}\label{eq:swarm}
    \BoldSt{\timeIndex}{}= \AugA{}{}\BoldSt{\timeIndex-1}{}+\AugB{}{}\ControlInput{\timeIndex-1}{},
\end{equation}
where $\BoldSt{\timeIndex}{}=[({\St{\timeIndex}{1}})^{T},({\St{\timeIndex}{2}})^{T},...,({\St{\timeIndex}{|\AgentSet{}|}})^{T}]^{T}$ is the vector of the states of the agents in the MAS at time step $\timeIndex$ and ${\St{\timeIndex}{\agtID}}=[\St{\timeIndex}{\agtID,\dimcounter_{1}},\St{\timeIndex}{\agtID,\dimcounter_{2}},...,\St{\timeIndex}{\agtID,\dimcounter_{|\mathcal{D}|}}]^{T}$ (where $\agtID\in{\AgentSet{}}$) denotes the state of agent $\agtID$ at time step $\timeIndex$; $\ControlInput{\timeIndex-1}{}=[({\NonBoldControlInput{\timeIndex-1}{1}})^{T},({\NonBoldControlInput{\timeIndex-1}{2}})^{T},...,({\NonBoldControlInput{\timeIndex-1}{|\AgentSet{}|}})^{T}]^{T}$ is the vector of the control inputs at time step $\timeIndex-1$ and ${\NonBoldControlInput{\timeIndex-1}{\agtID}}=[\NonBoldControlInput{\timeIndex-1}{\agtID,\dimcounter_{1}},\NonBoldControlInput{\timeIndex-1}{\agtID,\dimcounter_{2}},...,\NonBoldControlInput{\timeIndex-1}{\agtID,\dimcounter_{|\mathcal{D}|}}]^{T}$ is the control input vector of agent $\agtID$ at time step $\timeIndex-1$, and $\AugA{}{}$ and $\AugB{}{}$ are $|\AgentSet{}|\times|\AgentSet{}|$ diagonal time-invariant matrices. The dynamics equation of agent $\agtID$ can be expressed as $\St{\timeIndex}{\agtID}=\AugA{}{\agtID*}\St{\timeIndex-1}{\agtID}+\AugB{}{\agtID*}\NonBoldControlInput{\timeIndex-1}{\agtID}$, where $\St{\timeIndex}{\agtID}\in{\Env}$ and $\NonBoldControlInput{\timeIndex-1}{\agtID}\in\mathcal{U}=\{u|\Vert{u}\Vert_{\infty}\leq{u_{\text{max}}}\}$ for all $\agtID\in\AgentSet{}$ and for all $\timeIndex\in\distime$, and $\AugA{}{\agtID*}$ and $\AugB{}{\agtID*}$ refer to the $\agtID$-th row of matrix $\AugA{}{}$ and $\AugB{}{}$, respectively.
\begin{definition}\label{def:swarm traj}
We define the \textbf{system-level} trajectory $\ActGMNoT{}$ as a function $\ActGMNoT{}:\distime\rightarrow\Env$ to denote evolution of the average of the states of all the agents in the MAS within a finite time horizon defined in the discrete time domain $\distime$ and we define $\ActGM{\timeIndex}{}:=\frac{1}{|\AgentSet{}|}\sum\limits_{\agtID=1}^{|\AgentSet{}|}\St{\timeIndex}{\agtID}$. We also define the \textbf{agent-level} trajectory $\AgtTrajNoT{\agtID}$ as a function $\AgtTrajNoT{\agtID}:\distime\rightarrow\Env$ to denote the evolution of the state of each agent $\agtID$ within a finite time horizon defined in the discrete time domain $\distime$.
\end{definition}
In this paper, we represent the topology of the MAS with an \textit{undirected graph} $\graphG$ that is time-invariant.

\begin{definition}
We denote an undirected graph by $\graphG=(\nodeSet,\mathcal{E})$, where $\nodeSet=\{\Gennode_1,\Gennode_2,...,\Gennode_{\nodeNo}\}$ is a finite set of nodes,  $\mathcal{E}=\{e_{1,2},e_{1,3},...,e_{1,n_{\mathcal{E}}},e_{2,3},...,e_{n_{\mathcal{E}-1},n_{\mathcal{E}}}\}$ is a finite set of edges, and $n_{\mathcal{C}},n_{\mathcal{E}}\in\mathbb{N}=\{1,2,...\}$. In the set of edges $\mathcal{E}$, $e_{i,l}$ represents the edge that connects the nodes $\Gennode_{i}$ and $\Gennode_{l}$.
\end{definition}

Each node $\node_{i}$ of the undirected graph $\graphG$ represents an agent in the MAS. Each edge $e_{i,l}$ connecting the nodes $i$ and $l$ represents the fact that agnets $i$ and $l$ are neighbors, i.e., agent $i$ and $l$ communicate with each other. Hereafter, we denote the set of the neighboring agents of agent $\agtID$ with $\AgentSet{\agtID}$. Also, we denote the adjacency matrix of the graph $G$ with $\GraphMat$.

\subsection{Differential Privacy}\label{subsec: diff privacy}
In this subsection, we review the theoretical framework of differential privacy that we use in this paper \cite{KasraDiff}\cite{PappasDiff}\cite{XuDiff}. To apply differential privacy to protect the sensitive information of each agent, we use the ``input perturbation" approach, i.e., each agent adds noise to its state and then shares its noisy output with its neighboring agents. Before formalizing the definition of differential privacy, we explain the preliminary definitions and notations related to differential privacy.

For a trajectory $\kappa_i=\St{0}{\agtID},\St{1}{\agtID},...,\St{t}{\agtID},...$, where $\St{t}{\agtID}=\kappa_i[t]\in{\mathbb{R}}^{|\mathcal{D}|}$ for all $t$, we use the $\ell_p$ norm as $\Vert{\kappa_i}\Vert_{\ell_p}:=\left(\sum\limits^{\infty}_{t=1}\Vert{\St{\timeIndex}{\agtID}}\Vert^{p}_{p}\right)^{\frac{1}{p}}$, where $\Vert{.}\Vert_{p}$ is the ordinary $p$-norm on $\mathbb{R}^{d}$ ($d\in\nat\cup\{0\}$), and we define the set $\ell^{d}_p:=\{\kappa_{i}|\St{\timeIndex}{\agtID}\in{\mathbb{R}}^{d},\Vert{\kappa}\Vert_{\ell_p}<\infty\}$. Then, we define the truncation operator $\mathpzc{P}_{Q}$ over trajectories $\kappa_i$ as follows: $\mathpzc{P}_{Q}(\kappa_{i})=\St{\timeIndex}{\agtID}$ if $t\leq{Q}$; and $\mathpzc{P}(\kappa_{i})=0$, otherwise. Now, we define the set ${\Tilde{\ell}}^{|\mathcal{D}|}_{2}$ as the set of sequences of vectors in $\mathbb{R}^{|\mathcal{D}|}$ whose finite truncations are all in $\ell^{|\mathcal{D}|}_{2}$ ($p=2$ and $d=|\mathcal{D}|$). In other words, $\kappa_{i}\in{\Tilde{\ell}}^{|\mathcal{D}|}_{2}$ if and only if $\mathpzc{P}_{Q}(\kappa_{i})\in{{\ell}^{|\mathcal{D}|}_{2}}$ for all $Q\in\nat$.
\begin{definition}
\textbf{(Adjacency)} For a fixed adjacency parameter $\nu_i$, the adjacency relation $\mathpzc{A}_{\nu_i}$, for all $\kappa_{i},\kappa^{\prime}_{i}\in{{\Tilde{\ell}}^{|\mathcal{D}|}_{2}}$, is defined as Eq. \eqref{eq:adjan}. 
\end{definition}
\vspace{-20pt}
\begin{align}
 {\mathpzc{A}_{\nu_i}(\kappa_i,\kappa^{\prime}_{i})}
 &=
  \begin{cases}
   1, &\text{if}~ \Vert{\kappa_i-\kappa^{\prime}_{i}}\Vert_{\ell_2}\leq{\nu_i}\\
    0, & \textrm{otherwise}.
 \end{cases}
  \label{eq:adjan}
\end{align}
In other words, two trajectories generated by agent $\agtID$ are adjacent if the $\ell_2$ distance between the two is less or equal to $\nu_i$. Differential privacy is expected to make agent $i$'s true trajectory, denoted by $\kappa_i$, indistinguishable from all other trajectories contained in an $\ell_2$-ball of radius $\nu_i$ centred at the true trajectory $\kappa_i$. 

We use a probability space $(\Omega, \mathcal{F}, \mathbb{P})$ in order to state a formal definition for differential privacy for dynamical systems which specifies the probabilistic guarantees of privacy. For the formal definition of differential privacy mechanism that we explain shortly, we assume that the outputs of the mechanism is in ${\Tilde{\ell}}^{q}_{2}$ and uses a $\sigma$-algebra over ${\Tilde{\ell}}^{q}_{2}$ which is denoted by $\Theta^{q}_{2}$ and the construction of $\Theta^{q}_{2}$ is explained in \cite{hajek_2015}.
\begin{definition}\label{def:diff}
\textbf{(($\DiffEps{}{i},\DiffDelta{}{i}$)-Differential Privacy for Agent $i$)} With $\DiffEps{}{i}>0$ and $\DiffDelta{}{i}\in(0,\frac{1}{2})$ for agent $i$, a mechanism $\mathcal{M}:{\Tilde{\ell}}^{|\mathcal{D}|}_{2}\times{\Omega}\rightarrow{{\Tilde{\ell}}^{q}_{2}}$ is ($\DiffEps{}{i},\DiffDelta{}{i}$)-differentially private if for all adjacent trajectories $\kappa_{i},\kappa^{\prime}_{i}\in{{\Tilde{\ell}}^{|\mathcal{D}|}_{2}}$ and for all $S\in\Theta^{q}_{2}$, we have the following (Eq. \eqref{eq:diff}).   
\end{definition}
\vspace{-5pt}
\begin{equation}\label{eq:diff}
   \mathbb{P}[\mathcal{M}(\kappa_i)\in{S}]\leq{e}^{\DiffEps{}{i}}\mathbb{P}[\mathcal{M}(\kappa^{\prime}_{i})\in{S}]+\DiffDelta{}{i} 
\end{equation}
At each time step $\timeIndex$, each agent $\agtID$ outputs the value $\AugC{}{i*}\y{t}{\agtID}$, where $\AugC{}{i*}$ is the $i$-th row of the $|\AgentSet{}|\times|\AgentSet{}|$ time-invariant diagonal matrix $\AugC{}{}$. At time step $t$, for protecting the privacy of agent $i$, noise must be added to its output $\y{t}{i}$ so an adversary can not infer agent $i$'s trajectory $\kappa_i[t]$ from its output $\y{t}{\agtID}$. Calibrating the level of noise is done using ``sensitivity" of an agent's output.
\begin{definition}\label{def:sensi}
(\textbf{Sensitivity for Input Perturbation Privacy}) The $\ell_2$-norm sensitivity of agent $i$'s output map is the greatest distance between two output trajectories ${\vartheta}_{i}=\y{0}{i},\y{1}{i},...,\y{t}{i},...$ and $\vartheta^{\prime}_{i}=\y{0}{i},\y{1}{i},...,\y{t}{i},...$ defined as Eq. \eqref{eq:senpert} for $\kappa_{i},\kappa^{\prime}_{i}\in{{\Tilde{\ell}}^{|\mathcal{D}|}_{2}}$.
\end{definition}
\vspace{-10pt}
\begin{equation}\label{eq:senpert}
    \Delta_{\ell_2}\vartheta_i:= \sup\limits_{\kappa_{i},\kappa^{\prime}_{i}}\Vert{\AugC{}{i*}\kappa_i-\AugC{}{i*}\kappa^{\prime}_{i}}\Vert_{\ell_2}
\end{equation}
We can use $|\AugC{}{i*}|\nu_i$ as an upper bound for $\Delta_{\ell_2}\vartheta_i$ \cite{PappasDiff}. One of the most well-known mechanisms to enforce differential privacy is the Gaussian mechanism which requires adding Gaussian noise to the outputs of agent $i$.
\begin{lemma}\label{lem: gaussian}
\textbf{(Input Perturbation Gaussian Mechanism for Linear Systems)} For agent $i$ with trajectory $\kappa_{i}\in{\Tilde{\ell}}^{|\mathcal{D}|}_{i}$, we have the privacy parameters $\DiffEps{}{i}>0$ and $\DiffDelta{}{i}\in(0,\frac{1}{2})$, output trajectory $\vartheta_i\in{\Tilde{\ell}}^{q}_{i}$, and we denote agent $i$'s $\ell_2$-norm sensitivity by $\Delta_{\ell_2}\vartheta_i$. The Gaussian mechanism for ($\DiffEps{}{i},\DiffDelta{}{i}$)-differential privacy is defined as Eq. \eqref{eq:noisyy}.
\begin{equation}\label{eq:noisyy}
    \Tilde{y}_i[t]=y_i[t]+v_i[t],
\end{equation}
where $v_i[t]$ is a stochastic process with $v_i[t]\sim{\mathcal{N}}\left(0,(\sigma_{i})^2I_{q}\right)$, $I_q$ is a $q\times{q}$ identity matrix, and 
\begin{equation}
    \sigma_i\geq{\frac{}{}}\frac{\Delta_{\ell_2}\vartheta_i}{2\DiffEps{}{i}}\left(\iota_{\DiffDelta{}{}{i}}+\sqrt{(\iota_{\DiffDelta{}{i}})2\DiffEps{}{i}}\right)~\text{with}~ \iota_{\DiffDelta{}{i}}=:\mathpzc{G}^{-1}(\DiffDelta{}{i}),
\end{equation}
where $\mathpzc{G}(y):=\frac{1}{\sqrt{2\pi}}\int^{\infty}_{y}e^{-\frac{z^2}{2}}dz$.
\end{lemma}
\begin{proof}
See \cite{PappasDiff}.
\end{proof}
More clearly, the Gaussian mechanism adds i.i.d Gaussian noise point-wise in time to the output of agent $i$ to keep its state private. In this paper, we assume that the Gaussian noise $v_i$ is time-invariant. Also, we denote the vector of the noisy outputs of all the $|\AgentSet{}|$ agents at time step $t$ with $\BoldNoisyy{\timeIndex}{}$. In addition in this paper, we apply the differential privacy mechanism to the finite trajectory $\AgtTrajNoT{i}$ for each agent $i$.

\subsection{Metric Temporal Logic}
\label{metric temporal logic}
In this subsection, we briefly review the metric temporal logic (MTL) \cite{NasimUncer}.
We start with the Boolean semantics of MTL.
The domain $\mathbb{B}=\{True,False\}$ is the Boolean domain.
Moreover, we introduce a set $\varPi$ which is a set of \textit{atomic predicates} each of which maps $\Env$ to $\mathbb{B}$.
Each of these predicates can hold values \textit{True} or \textit{False}.
The syntax of MTL is defined recursively as follows.
\begin{align*}
    \formula{}{} &:=
    \top
    \mid \pi
    \mid \lnot\formula{}{}
    \mid \formula{}{1}\land\formula{}{2}
    \mid \formula{}{1}\lor\formula{}{2}
    \mid \formula{}{1}\luntil_{I}\formula{}{2}
\end{align*}

\noindent where $\top$ stands for the Boolean constant \textit{True},
$\pi$ is an atomic predicate such that $\pi\in\varPi$.
$\lnot$ (negation), $\land$ (conjunction), $\lor$ (disjunction) 
are standard Boolean connectives, and ``$\luntil$'' is the temporal operator ``until''.
We add syntactic sugar, and introduce the temporal operators ``$\leventually$'' and ``$\lglobally$'' representing ``{eventually}'' and ``always'', respectively.
$I$ is a time interval of the form $I=[a,b)$, where $\textcolor{black}{a < b}$, and they are non-negative integers. We define the set of the states that satisfy $\pi$ as $\mathcal{O}(\pi)\subset{\Env}$.

We denote the distance from $\AgtTrajNoT{}$ to a set $\mathcal{J}\subseteq{\Env}$ as $\mathbf{dist}_{{f}}(\AgtTrajNoT{},\mathcal{J}):=\inf\{{f}({\AgtTrajNoT{}},s^{\prime})|s^{\prime}\in{cl}(\mathcal{J})\}$, where $f$ is a metric on $\Env$, and $cl(\mathcal{J})$ denotes the closure of the set $\mathcal{J}$. In this paper, we use the metric ${f}(\AgtTrajNoT{},s^{\prime})=\Vert{\AgtTrajNoT{}-s^{\prime}}\Vert_{2}$, where $\Vert{.}\Vert_{2}$ denotes the 2-norm. We denote the depth of $\AgtTrajNoT{}$ in $\mathcal{J}$ by $\mathbf{depth}_{{f}}(\AgtTrajNoT{},\mathcal{J}):=\mathbf{dist}_{{f}}(\AgtTrajNoT{},\Env\setminus\mathcal{J})$. We define the signed distance from $\AgtTrajNoT{}$ to $\mathcal{J}$ as $\mathbf{Dist}_{{f}}(\AgtTrajNoT,\mathcal{J}):=-\mathbf{dist}_{f}(\AgtTrajNoT{},\mathcal{J})$, if $\AgtTrajNoT{}\not\in{\mathcal{J}}$; and $\mathbf{Dist}_{{f}}(\AgtTrajNoT{},\mathcal{J}):=\mathbf{depth}_{{f}}(\AgtTrajNoT{},\mathcal{J})$ if $\AgtTrajNoT{}\in{\mathcal{J}}$ \cite{XuDiff}.

\begin{definition}
The \textbf{minimum necessary length} of an MTL formula $\formula{}{}$, denoted by $\horizon(\formula{}{})$, is the minimum time steps required to evaluate the truth value of $\formula{}{}$.
\end{definition}

\begin{definition}
The Boolean semantics of an MTL formula $\formula{}{}$ with the necessary length of $\horizon(\formula{}{})$, for a trajectory $\AgtTrajNoT{}$ at time step $t$ is defined recursively as follows.
\[
\begin{split}
(\AgtTrajNoT{},t)\models\pi~\mbox{iff}&~~
t\leq{\horizon(\phi)}~\mbox{and}~
\AgtTraj{t}{}\in\mathcal{O}(\pi)\\
(\AgtTrajNoT{},t)\models\lnot\formula{}{}~\mbox{iff}&~~(\AgtTrajNoT{},t)\not\models\formula{},\\
(\AgtTrajNoT{},t)\models\formula{}{1}\wedge\formula{}{2}~\mbox{iff}&~~(\AgtTrajNoT{},t)\models\formula{}{1}~\mbox{and}\\
&~~(\AgtTrajNoT{},t)\models\formula{}{2},\\
(\AgtTrajNoT{},t)\models\formula{}{1}\luntil_{[a,b)}\formula{}{}~\mbox{iff}&~~\exists{t'}\in[t+a,t+b),\\
~(\AgtTrajNoT{},t') \models \formula{}{2} ~\mbox{and} &~~\forall{t''}\in{[t+a,t')}\mbox{,}
~(\AgtTrajNoT{},t'') \models \formula{}{1}.
\end{split}
\]
\end{definition}

Robust semantics quantifies the degree at which a certain trajectory satisfies or violates an MTL formula $\formula{}{}$ at time step $t$. The robustness degree of a an MTL formula
$\formula{}{}$ with respect to a trajectory $\AgtTrajNoT{}$ at time step $t$ is given by $\robustness(\AgtTrajNoT{},\formula{}{},t)$, where $\robustness(\AgtTrajNoT{},\formula{}{},t)$ can be calculated
recursively via the robust semantics as follows. 
\vspace{-5pt}
\begin{align*}
    \begin{split}
    \robustness(\AgtTrajNoT{},\pi,t)  &= \mathbf{Dist}_{{f}}(\AgtTraj{\timeIndex}{},\mathcal{O}(\pi)),\\
    \robustness(\AgtTrajNoT{},\lnot\formula{},t)  &= -\robustness(\AgtTrajNoT{},\formula{}{},t),\\
    \robustness(\AgtTrajNoT{},\formula{}{1}\wedge\formula{}{2},t)  &= \min(\robustness(\AgtTrajNoT{},\formula{}{1},t),\robustness(\AgtTrajNoT{},\formula{}{2},t)),\\
    \robustness(\AgtTrajNoT{},\formula{}{1}\luntil_{\lbrack{a},b)}\formula{}{2},t) &=
    \max\limits_{t+a \leq t' < t+b}(  \min( \robustness(\AgtTrajNoT{},\formula{}{2},t'), \\
    &\min\limits_{t+a \leq t'' < t'} \robustness(\AgtTrajNoT{},\formula{}{1},t'') )).
    \end{split}
\end{align*}

\subsection{Estimation of the states Using Kalman Filter}\label{subsec: kalman}
In this subsection, we review the Kalman filter equations that are used to calculate the optimal estimates of the states using given noisy outputs. 
The Kalman filter equations are as follow \cite{Kalman}.
\vspace{-5pt}
\begin{align}
\label{eq:Est}
    \EstSt{\timeIndex}{}&=\AugA{}{}\EstSt{\timeIndex-1}{}+\AugB{}{}\ControlInput{\timeIndex-1}{}\\
    &~~+\KalmanGain{\timeIndex}{}(\BoldNoisyy{\timeIndex}{}-\AugA{}{}\EstSt{\timeIndex-1}{}-\AugB{}{}\ControlInput{\timeIndex-1}{}),
    \\
    \label{eq:KalmanGain}
    \KalmanGain{\timeIndex}{}&=\StCovarianceMat{\timeIndex-1}{}(\StCovarianceMat{\timeIndex-1}{}+\NoiseCo{}{})^{-1},\\
    \label{eq:ErrorCo}
    \StCovarianceMat{\timeIndex}{}&=(I_{|\AgentSet{}|}-\KalmanGain{\timeIndex}{})\StCovarianceMat{\timeIndex-1}{},
    \end{align}
where $\EstSt{\timeIndex}{}=[({\EstStAgt{\timeIndex}{1}})^{T},({\EstStAgt{\timeIndex}{2}})^{T},...,({\EstStAgt{\timeIndex}{|\AgentSet{}|}})^{T}]^{T}$ is the vector of the estimated states of $|\AgentSet{}|$ agents at time step $\timeIndex$ and ${\EstStAgt{\timeIndex}{\agtID}}=[\EstStAgt{\timeIndex}{\agtID,\dimcounter_{1}},\EstStAgt{\timeIndex}{\agtID,\dimcounter_{2}},...,\EstStAgt{\timeIndex}{\agtID,\dimcounter_{|\mathcal{D}|}}]^{T}$ (where $\agtID\in{\AgentSet{}}$) denotes the $|\mathcal{D}|$-dimensional estimated state of agent $\agtID$ at time step $\timeIndex$, $\KalmanGain{\timeIndex}{}$ is the Kalman gain matrix at time step $\timeIndex$ and is a $|\AgentSet{}|\times|\AgentSet{}|$ matrix, $\StCovarianceMat{\timeIndex}{}$ is the covariance matrix of state estimation error at time step $\timeIndex$ and is a $|\AgentSet{}|\times|\AgentSet{}|$ matrix, $\NoiseCo{}{}=\mathbb{E}(\noise{}{}{\noise{}{}}^T)$ is the covariance matrix of the noise vector $\noise{}{}$. Also, we assume that each agent $\agtID$ knows that $\noise{}{}$ conforms to a Gaussian distribution with 0 mean and covariance matrix $\NoiseCo{}{}$. Also, each agent $\agtID$ knows that $\mathbb{E}(\Vert{\noise{}{}}\Vert^2)\leq{|\AgentSet{}|v_{\text{max}}}$ and $\mathbb{E}(\Vert{\EstStAgt{0}{\agtID}}-\St{0}{\agtID}\Vert^2)\leq{\StNoT{\textrm{max}}}$ with $v_{\text{max}}$ and ${\StNoT{\textrm{max}}}$ being arbitrary values \cite{RuixuanDistributed}.
 
\begin{table}[hbt!]
\addtolength{\tabcolsep}{-5pt}
\begin{center}
\begin{tabular}{ c | c }
\toprule
 \textbf{Notation} & \textbf{{Definition}} \\
 \hline
 $\AgtTraj{t}{i}$ & actual agent-level trajectory \\
 & of agent $i$ at time step $t$
 \\
  \hline
  $\AgtTraj{t}{i,d_m}$ & $m$-th dimension of 
  \\
  & actual agent-level trajectory \\
 & of agent $i$ at time step $t$
 \\
 \hline
 $\BoldSt{\timeIndex}{}$ & $|\AgentSet{}|\times{1}$ vector containing the actual    \\ 
 & states of agents at time step $\timeIndex$\\
 \hline
  $\ActGM{t}{}$ & actual system-level   \\
  & trajectory at time step $t$
  \\
  \hline
  $\EstStAgt{\timeIndex}{\agtID}$ & estimated state of agent $\agtID$ at time step $\timeIndex$\\
  \hline
 $\EstSt{\timeIndex}{}$ & $|\AgentSet{}|\times{1}$ vector containing the estimated    \\
 &states of the agents at time step $t$
 \\
 \hline
 $\EstSt{\timeIndex}{\agtID}$ &$|\AgentSet{}|\times{1}$ vector containing the estimated states of the  \\
 & agents estimated by agent $i$ at time step $t$
 \\
 \hline
 $\GMnonBold{\timeIndex}{\agtID}$ & estimated system-level trajectory
 \\
  & estimated by agent $i$ at time step $t$
  \\
 \hline
 $\GMnonBold{\timeIndex}{\agtID,d_m}$ 
 & $m$-th dimension of \\
 & the estimated system-level trajectory
  \\
  & estimated by agent $i$ at time step $t$
  \\
 \hline
 $\GMori{\timeIndex}{}$ & $|\AgentSet{}|\times{1}$ vector containing the estimated system-level \\ &
 trajectories at time step $t$ estimated by $|\AgentSet{}|$ agents
 \\
 \hline
 $\GM{\timeIndex}{\agtID}$ & $|\AgentSet{}|\times{1}$ vector containing the estimate of $\GMori{\timeIndex}{}$\\
 & estimated by agent $i$ at time step $t$
 \\
 \hline
 $\Noisyy{\timeIndex}{\agtID}$ & noisy output of
 agent $i$ at time step $t$
 \\
 \hline
 $\Boldy{t}{}$ & $|\AgentSet{}|\times{1}$ vector containing the outputs \\
 & of the agents at time step $t$
 \\
 \hline
 $\BoldNoisyy{t}{i}$ & $|\AgentSet{}|\times{1}$ vector containing the noisy outputs \\
 & of the agents received through\\ & communication or calculated by agent $i$
\\
 $\NonBoldControlInput{\timeIndex}{\agtID}$& control input of agent $i$ at time\\
 &  step $t$ calculated by agent $i$
 \\
 \hline
 $\ControlInput{\timeIndex}{}$ & $|\AgentSet{}|\times{1}$ vector containing the control inputs \\
 &  of $|\AgentSet{}|$ agents at time step $t$\\
 \hline
 $\ControlInput{\timeIndex}{\agtID}$ & $|\AgentSet{}|\times{1}$ vector containing the control inputs \\
 & calculated by agent $i$ at time step $t$\\
 \bottomrule
\end{tabular}
\caption{{List of important notations.}}
\label{table:1}
\end{center}
\end{table}

\section{Estimation of system-level Trajectory in Multi-Agent Systems with MTL Specifications}\label{sec: Est of robot swarm}
In this section, we review a method to estimate the system-level trajectory $\ActGMNoT{}$ in a situation where each agent $\agtID$ shares only its noisy outputs $\Tilde{y}_{\agtID}$ with its neighbors in an MAS. In this MAS, each agent $\agtID$ has \textit{asynchronous communication} with only its neighboring agents, i.e., only two agents can communicate at each time step $\timeIndex$ and the probability of each agent $\agtID$ being \textit{active} at time step $\timeIndex$ is $\frac{1}{|\AgentSet{}|}$ (here active means agent $\agtID$ can initiate communication with another agent). In other words, if agent $\agtID$ is not active at time step $\timeIndex$, then it can not initiate communication with its neighboring agent.

Intuitively, in order to collaborate in satisfying a system-level task $\formula{}{\textrm{s}}$, each agent $\agtID$ needs to have access to the system-level trajectory $\ActGMNoT{}$. However, in the situation where each agent $\agtID$ has only access to the noisy outputs of its neighboring agents, each agent $\agtID$ needs to have an estimate of the system-level trajectory $\ActGMNoT{}$ while taking into consideration that the probability of the satisfaction of the MTL specification $\formula{}{\textrm{s}}$ is higher than a given minimum value $\MinConf$.

In what follows, we review a method by which each agent can estimate the system-level trajectory $\ActGMNoT{}$ in a distributed manner \cite{RuixuanDistributed}. The main idea in this method is that each agent is able to compute an estimate of the actual system-level trajectory $\ActGMNoT{}$ such that the estimation error converges to zero when the time $\timeIndex$ goes to infinity. We assume that at time step $\timeIndex-1$, agents $\agtID$ and $\counterl$ have computed the estimated system-level trajectory $\GMnonBold{\timeIndex-1}{\agtID}$ and $\GMnonBold{\timeIndex-1}{\counterl}$, respectively. At time step $\timeIndex$, agents $\agtID$ and $\counterl$ communicate with each other and update their estimates of the system-level trajectory $\ActGM{\timeIndex}{}$ using Eqs. \eqref{eq:Est1} and \eqref{eq:Est2}, respectively.
\vspace{-10pt}
\begin{equation}\label{eq:Est1}
    \GMnonBold{\timeIndex}{\agtID}=\frac{1}{2}(\GMnonBold{\timeIndex-1}{\agtID}+\GMnonBold{\timeIndex-1}{\counterl})+\EstStAgt{\timeIndex}{\agtID}-\EstStAgt{\timeIndex-1}{\agtID},
\end{equation}
\begin{equation}\label{eq:Est2}
    \GMnonBold{\timeIndex}{\counterl}=\frac{1}{2}(\GMnonBold{\timeIndex-1}{\agtID}+\GMnonBold{\timeIndex-1}{\counterl})+\EstStAgt{\timeIndex}{\counterl}-\EstStAgt{\timeIndex-1}{\counterl},
\end{equation}
and other agents update their estimates of the system-level trajectory using Eq. \eqref{eq:EstAll}
\begin{equation}\label{eq:EstAll}
    \GMnonBold{\timeIndex}{\counterk}=\GMnonBold{\timeIndex-1}{\counterk}+\EstStAgt{\timeIndex}{\counterk}-\EstStAgt{\timeIndex-1}{\counterk},~\counterk\in\AgentSet{}~\text{and} ~\counterk\neq{\agtID,\counterl}.
\end{equation}
We can reformulate the equations of the update of the estimated system-level trajectory by the agents in the vector form using Eq. \eqref{eq:EstVector}
\begin{equation}\label{eq:EstVector}
    \GMori{\timeIndex}{}=\TopMatrix\GMori{\timeIndex-1}{}+\EstSt{\timeIndex}{}-\EstSt{\timeIndex-1}{},
\end{equation}
where $\GMori{\timeIndex}{}=[(\GMnonBold{\timeIndex}{1})^T,(\GMnonBold{\timeIndex}{2})^T,...,(\GMnonBold{\timeIndex}{|\AgentSet{}|})^T]^{T}$ is the vector containing the estimates of the actual system-level trajectory $\ActGM{\timeIndex}{}$ made by the $|\AgentSet{}|$ agents at time step $\timeIndex$, and $\GMnonBold{\timeIndex}{\agtID}=[\GMnonBold{\timeIndex}{\agtID,d_1},...,\GMnonBold{\timeIndex}{\agtID,d_{|\mathcal{D}|}}]^T$; matrix $\TopMatrix$ is the system matrix to update $\GMori{\timeIndex}{}$ and is obtained by solving the following optimization problem \cite{BoydGossip}\cite{RuixuanDistributed}.
\vspace{-5pt}
\begin{align}\label{prob:optmizeV}
\argmin\limits_{{\TopMatrix}}~~~&g,\\
\text{subject to}~~~&\ProbMatrix{il}\geq{0},~\ProbMatrix{i,l}={0},~\text{if}~e_{i,l}\not\in{\mathcal{E}},\\
&\TopMatrix=\frac{1}{|\AgentSet{}|}\sum\limits_{i=1}^{|\AgentSet{}|}\sum\limits_{l=1}^{|\AgentSet{}|}\ProbMatrix{il}(\Tilde{\TopMatrix})_{il},~\TopMatrix-\frac{1}{|\AgentSet{}|}\boldsymbol{1}\boldsymbol{1}^{T}\preceq{g}I_{|\AgentSet{}|},\\
&\sum\limits_{j=1}^{|\AgentSet{}|}\ProbMatrix{il}=1, \forall{i}.
\end{align}
In the optimization problem \eqref{prob:optmizeV}, $\ProbMatrix{}$ is a  $|\AgentSet{}|\times{|\AgentSet{}|}$ matrix  where each entry $\ProbMatrix{\agtID\counterl}$ represents the probability that agent $\agtID$ communicates with agent $\counterl$. As we mentioned earlier, the probability of agent $\agtID$ to be active is $\frac{1}{|\AgentSet{}|}$; thus matrix $\TopMatrix$ has a probability of $\frac{1}{|\AgentSet{}|}\ProbMatrix{il}$ to be equal to $(\Tilde{\TopMatrix})_{\agtID\counterl}=I_{|\AgentSet{}|}-\frac{(e_{\agtID}-e_{\counterl})(e_{\agtID}-e_{\counterl})}{2}$, where $I_{|\AgentSet{}|}$ is a $|\AgentSet{}|\times{}|\AgentSet{}|$ identity matrix and $e_{\agtID}=[0,...,1,...,0]^{T}$ is a $|\AgentSet{}|\times{1}$ vector with the $\agtID$-th entry to be 1 and zero in all other entries \cite{BoydGossip}\cite{RuixuanDistributed}. It can be shown that the expected value of the matrix $\TopMatrix[\timeIndex]$, denoted by $\mathbb{E}(\TopMatrix[\timeIndex])$, is constant at different time steps $\timeIndex$ \cite{BoydGossip}. $\boldsymbol{1}$ a is $|\AgentSet{}|\times{1}$ vector with all the entries to be 1, and $\mathcal{E}$ is the set of the edges of the undirected graph $G$ of the MAS.


In order to use $\GMori{\timeIndex}{}$ for synthesizing controller inputs for the MAS, we need to address two important matters: (1) it is crucial to guarantee that the estimation error of the vector of the estimated system-level trajectories $\GMori{\timeIndex}{}$ (in comparison with the actual system-level trajectory $\ActGM{\timeIndex}{}$) converges to zero when $\timeIndex\rightarrow\infty$, and (2) we need to guarantee that the actual system-level trajectory $\ActGMNoT{}$ satisfies $\formula{}{\textrm{s}}$ with a probability higher than a minimum value $\MinConf$ given that agents do not have access to the actual system-level trajectory. Therefore, we need to provide the guarantee of $\ActGM{\timeIndex}{}$ satisfying $\formula{}{\textrm{s}}$ using the vector of the estimated system-level trajectories $\GMori{\timeIndex}{}$. In what follows, we provide Theorem \ref{the:EstErr} and Lemma \ref{lem:MonRul} that address the two mentioned issues, respectively \cite{RuixuanDistributed}. 
\vspace{-7pt}
\begin{theorem}\label{the:EstErr}
The estimation error $\mathbb{E}(\Vert{\GMori{\timeIndex}{}-\ActGM{\timeIndex}{}\boldsymbol{1}}\Vert_{\infty})$ converges to zero when $\timeIndex\rightarrow\infty$ for an MAS consisting of $|\AgentSet{}|$ agents, where $\GMori{\timeIndex}{}$ is the vector of estimates of the actual system-level trajectory $\ActGM{\timeIndex}{}$ at time step $\timeIndex$ calculated by $|\AgentSet{}{}|$ agents. 
\end{theorem}
\vspace{-16pt}

\begin{proof}
See \cite{RuixuanDistributed}
\end{proof}

\vspace{-7pt}

We can provide an upper bound for the estimation error $\mathbb{E}(\Vert{\GMori{\timeIndex}{}-\ActGM{\timeIndex}{}\boldsymbol{1}}\Vert_{\infty})$ using Corollary \ref{col:errorbound}. Here we assume that each agent knows $\mathbb{E}(\Vert{\GM{0}{}-\overline{{\zeta}}[{0}]\boldsymbol{1}}\Vert_{\infty})\leq{\zeta_{\text{max}}}$, where $\overline{{\zeta}}[0]=\frac{1}{|\AgentSet{}|}\sum\limits^{|\AgentSet{}|}_{\agtID=1}\zeta_{\agtID}[{0}]$ and $\zeta_{\text{max}}$ is an arbitrary value.

\vspace{-5pt}
\begin{corollary}\label{col:errorbound} 
For the estimation error $\mathbb{E}(\Vert{\GMori{\timeIndex}{}-\ActGM{\timeIndex}{}\boldsymbol{1}}\Vert_{\infty})$ at time step $\timeIndex$, we have the upper bound defined as
\vspace{-12pt}
    
\begin{align}\label{eq:errorbound}
    \EstErrorBound{\timeIndex}{}=\lambda^{\timeIndex}\sqrt{|\AgentSet{}|}\zeta_{\text{max}}\mathcal{L}_1\sum\limits^{\timeIndex}_{\counterk=1}\lambda^{\timeIndex-\counterk}~~~~~~~~~~~~~~~~~~~~~~~~~~~\\
    \sqrt{\delta_{\text{max}}(\counterk)+\delta_{\text{max}}(\counterk-1)+2|\AgentSet{}|(u_{\text{max}})^2}+\mathcal{L}_2\sqrt{\delta_{max}({\timeIndex})},
\end{align}
where $\lambda$ is the second largest eigenvalue of matrix $\TopMatrix$, $\mathcal{L}_1$ and $\mathcal{L}_2$ are two Lipschitz constants and $\delta_{\textrm{max}}(\timeIndex)=\frac{|\AgentSet{}|^2s_{\textrm{max}}v_{\textrm{max}}}{v_{\textrm{max}}+\timeIndex{s_{\textrm{max}}}}$.
\end{corollary}
For further details regarding the derivation of the upper bound $\EstErrorBound{\timeIndex}{}$, we kindly refer the reader to Corollary 4 in  \cite{RuixuanDistributed}.

We use the upper bound $\EstErrorBound{\timeIndex}{}$ to provide the guarantee that the probability of $\ActGM{\timeIndex}{}$ satisfying the system-level MTL specification $\formula{}{\textrm{s}}$, referred to as \textit{confidence level}, is higher than a minimum value $\MinConf$; therefore, we provide a set of constraints that each agent $\agtID$ must satisfy recursively \cite{RuixuanDistributed}. 
\begin{lemma}\label{lem:MonRul}
Let $\MinRob\geq{0}$ denote the minimum required robustness degree of $\GMoriNoT{}$ satisfying a given system-level MTL specification $\formula{}{}$ at time step $\timeIndex$, the confidence level of agent $\agtID$ of $\ActGM{\timeIndex}{}$ satisfying $\formula{}{}$ at time step $\timeIndex$, denoted by $\ConfidenceLe{}{\agtID}((\ActGMNoT{},\timeIndex)\models{\formula{}{}})$, must satisfy the following constraints. 
\begin{align*}
   \ConfidenceLe{}{\agtID}((\ActGMNoT{}{}{},\timeIndex)\models{\pi})&\geq{1}-\frac{\EstErrorBound{\timeIndex}{}}{\MinRob},\\
   \ConfidenceLe{}{\agtID}((\ActGMNoT{},\timeIndex)\models{\formula{}{1}\land{\formula{}{2}}})&\geq{\ConfidenceLe{}{\agtID}((\ActGMNoT{},\timeIndex)\models{\formula{}{1}})}\\
   &+{\ConfidenceLe{}{\agtID}((\ActGMNoT{},\timeIndex)\models{\formula{}{2}})}-1,\\
    \ConfidenceLe{}{\agtID}((\ActGMNoT{},\timeIndex)\models{\formula{}{1}\lor{\formula{}{2}}})&\geq{1-\min\{1-\ConfidenceLe{}{\agtID}((\ActGMNoT{},\timeIndex)\models{\formula{}{1}})},\\
    &1-{\ConfidenceLe{}{\agtID}((\ActGMNoT{},\timeIndex)\models{\formula{}{2}})}\},\\
    {\ConfidenceLe{}{\agtID}((\ActGMNoT{},\timeIndex)\models{\lglobally_{[a,b]}\formula{}{}})}&\geq{1-\min\limits_{t'\in[t+a,t+b]}\{1-{\ConfidenceLe{}{\agtID}((\ActGMNoT{},t')\models{\formula{}{}})}},\\
    {\ConfidenceLe{}{\agtID}((\ActGMNoT{},\timeIndex)\models{\leventually_{[a,b]}\formula{}{}})}&\geq{1-\min\limits_{t'\in[t+a,t+b]}\{1-{\ConfidenceLe{}{\agtID}((\ActGMNoT{},t')\models{\formula{}{}})}\}},\\
    {\ConfidenceLe{}{\agtID}((\ActGMNoT{},\timeIndex)\models{\formula{}{1}\luntil_{[a,b]}\formula{}{2}})}&\geq{1-\min\limits_{t'\in[t+a,t+b]}\{1-{\ConfidenceLe{}{\agtID}((\ActGMNoT{},t')\models{\formula{}{2}}})}\\
    &+(\sum\limits^{t'}_{t''=t+a}1-\ConfidenceLe{}{i}((\ActGMNoT{},t'')\models{\formula{}{1}})\}.
\end{align*}

\end{lemma}
\begin{proof}
The listed constraints can be proven using Markov inequality. For further details, see the proof of Lemma 5 in \cite{RuixuanDistributed}.
\end{proof}

\vspace{-20pt}
\section{Problem Formulation}\label{Problem Formulation}

In this section, we formalize the problem of synthesizing controller inputs for an MAS consisting of $|\AgentSet{}|$ agents in a differentially private manner, where agents are required to collaborate to satisfy a system-level task specified using an MTL specification $\formula{}{\textrm{s}}$ with a minimum probability $\MinConf$, and at the same time, each agent $\agtID$ should satisfy an agent-level MTL specification $\formula{}{\agtID}$. In this MAS, each agent $\agtID$ communicates its noisy output $\Noisyy{\timeIndex}{\agtID}$ with its neighbors in order to keep its actual state $\St{\timeIndex}{\agtID}$ private from its neighboring agents while agent $\agtID$ is aware of its own actual state $\St{\timeIndex}{\agtID}$, and the communication is asynchronous.

 We want to synthesize the controller inputs in a distributed manner, i.e., each agent $\agtID$ synthesizes its own controller input $u_{\agtID}[{\timeIndex}]$ at time step $\timeIndex$. Hence, agent $\agtID$ must calculate $\GMori{\timeIndex}{}$ for computing the controller input $u_{\agtID}[{\timeIndex}]$ while taking into consideration that the probability of the satisfaction of the MTL specification  $\formula{}{\text{s}}$ by the actual system-level trajectory $\ActGM{\timeIndex}{}$ is higher than a minimum value. Hereafter, we denote the vector of the estimated system-level trajectories computed by agent $\agtID$ using $\GM{\timeIndex}{\agtID}$.

 \begin{remark}
At each time step $\timeIndex$, each agent $\agtID$ computes its own estimate of $\ActGM{\timeIndex}{}$, denoted by $\GMnonBold{\timeIndex}{\agtID}$, in addition to computing the estimates of other agents $l\in{\AgentSet{}}-\{\agtID\}$ from $\ActGM{\timeIndex}{}$. Thus, $\GM{\timeIndex}{\agtID}$ represents the vector containing the estimate of agent $\agtID$ from $\ActGM{\timeIndex}{}$, denoted by $\zeta_{i}[t]$, and the estimates of other agents from $\ActGM{\timeIndex}{}$ computed by agent $\agtID$.
\end{remark}

\begin{remark}\label{re:graph}
 In the MAS with $|\AgentSet{}|$ agents, each agent $i$ (1) has access to the time-invariant graph-structure of the MAS givn by the undirected graph $G$ and (2) knows the fact that the asynchronous communication of agent $\agtID$ is only with its neighboring agents $\counterl\in\AgentSet{i}$.
\end{remark}

Now, we formalize the problem of synthesizing control inputs for the control horizon of $2\horizon$ for a MAS with $|\AgentSet{}|$ agents in a distributed and differentially private manner. 
\begin{problem}\label{prob:Distcontroller}
Given an MAS consisting of $|{\AgentSet{}}|$ agents, the privacy parameters  $\DiffEps{}{\agtID}\in[\DiffEps{}{\textrm{min}},\DiffEps{}{\textrm{max}}]$ (where $0<\DiffEps{}{\textrm{min}}<\DiffEps{}{\textrm{max}}$) and $\DiffDelta{}{\agtID}\in[\DiffDelta{}{\textrm{min}},\DiffDelta{}{\textrm{max}}]$ (where $0<\DiffDelta{}{\textrm{min}}<\DiffDelta{}{\textrm{max}}<\frac{1}{2}$), and the objective function $\objective=\sum\limits_{\counterk=1}^{2\horizon}\Vert{\ControlInput{\counterk}{\agtID}}\Vert^2$, synthesize the controller inputs ${\ControlInput{\timeIndex}{\agtID}}$ in the control horizon $2\horizon$ in a distributed and differentially private manner such that the satisfaction of the system-level MTL specification $\formula{}{\textrm{s}}$ with the probability higher than $\MinConf$, i.e., $\ConfidenceLe{}{}((\ActGM{0:\horizon-1}{},0)\models{\formula{}{\textrm{s}}})>\MinConf$ and the satisfaction of the agent-level MTL specification $\formula{}{\agtID}$ by agent $\agtID$ is guaranteed while the objective function $\objective$ is minimized.
\end{problem}
\vspace{-8pt}
Hereafter, the subscript $\agtID$ in the variables $\EstSt{\timeIndex}{\agtID}$, $\ControlInput{\timeIndex}{\agtID}$, $\BoldNoisyy{\timeIndex}{\agtID}$ means that these variables have been calculated or received by agent $\agtID$ at time step $\timeIndex$.

\section{Distributed Differentially Private Receding Horizon Control for Multi-Agent Systems with MTL Specifications  }\label{sec:Implementation}
In this section, we introduce an approach for synthesizing control inputs for an MAS with MTL specifications in a distributed and differentially private manner. Based on the settings of Problem \ref{prob:Distcontroller}, at each time step $\timeIndex$, each agent $\agtID$ should synthesize its own control inputs in the time horizon $[\timeIndex,\timeIndex+2\horizon-1]$ for satisfying the system-level and agent-level specifications $\formula{}{\textrm{s}}$ and $\formula{}{\agtID}$.
For solving Problem \ref{prob:Distcontroller}, we adopt a receding horizon control (RHC) for synthesizing control inputs for satisfying MTL specifications $\formula{}{\textrm{s}}$ and $\formula{}{\agtID}$ while minimizing a given objective function $J$.

In RHC for satisfying  MTL specifications, we incorporate mixed-integer linear programming (MILP) to encode given MTL specifications as constraints in the optimization problem that is solved at each time step $\timeIndex$. In \cite{MPCDonze}, Raman \textit{et. al.} introduce a framework for encoding MTL specifications as MILP constraints, and we incorporate this framework in this paper. 

The proposed optimization formulation for synthesizing control inputs for an  MAS in a differentially private manner can be seen in \eqref{prob:optCon} referred to as $\MILP$.

\begin{align}
{\argmin\limits_{\ControlInput{\timeIndex:\timeIndex+2\horizon-1}{\agtID}}} ~ &
\sum_{\counterk=\timeIndex}^{\timeIndex+2\horizon-1}\Vert{\ControlInput{\counterk}{\agtID}}\Vert^2\end{align}
\begin{align}
    \text{subject to:}~\St{\counterk+1}{\agtID}=&\AugA{}{\agtID*}\St{\counterk}{\agtID}+\AugB{}{\agtID*}u_{\agtID}[{\counterk}],\\
    &\forall{\counterk\in\{\timeIndex,\timeIndex+1,...,\timeIndex+2\horizon-1\}},\\
    \EstSt{\counterk+1}{\agtID}=&\AugA{}{}\EstSt{\counterk}{\agtID}+\AugB{}{}\ControlInput{\counterk}{\agtID}+\KalmanGain{\counterk}{}
    \\
    &(\BoldNoisyy{\counterk+1}{\agtID}-\AugA{}{}\EstSt{\counterk}{\agtID}+\AugB{}{}\ControlInput{\counterk}{\agtID}),\\
    &\forall{\counterk\in\{\timeIndex,\timeIndex+1,...,\timeIndex+2\horizon-1\}},\\
    \GM{\counterk+1}{\agtID}=&\TopMatrix\GM{\counterk}{\agtID}+\EstSt{\counterk+1}{\agtID}-\EstSt{\counterk}{\agtID},\\
    &\forall{\counterk\in\{\timeIndex,\timeIndex+1,...,\timeIndex+2\horizon-1\}},\\
    \BoldNoisyy{\counterk+1}{\agtID}=&\AugC{}{}(\EstSt{\counterk}{\agtID}+\ControlInput{\counterk}{\agtID})
    \\
    &+\frac{\AugC{}{}\AugB{}{}\AugA{}{}}{2}(\ControlInput{\counterk}{\agtID}-\ControlInput{\counterk-1}{\agtID}),\\
    &\forall{\counterk\in\{\timeIndex,\timeIndex+1,...,\timeIndex+2\horizon-1\}},\\
    \robustness[\agtID](\AgtTraj{0:\horizon-1}{\agtID},{\formula{}{{\agtID}}},\STLtimeIndex)>&\PiCon{\STLtimeIndex}{},~\forall{{\STLtimeIndex}\in\{0,1,...,\horizon-1\}},\\
    \ConfidenceLe{}{\agtID}((\ActGM{0:\horizon-1}{\agtID},\STLtimeIndex)\models{\formula{}{\textrm{s}}})>&\Confidence{\STLtimeIndex}{},~\forall{{\STLtimeIndex}\in\{0,1,...,\horizon-1\}},\\
    u_{\text{min}}\leq\ControlInput{\counterk}{\agtID}\leq{u_{\text{min}}}~\forall&{{\counterk}\in\{0,1,...,2\horizon-1\}}.\\
    \label{prob:optCon}
\end{align}

A challenge in incorporating Eq. \eqref{eq:Est}, in the optimization formulation in $\MILP$, is to calculate the vector of the noisy output  $\BoldNoisyy{\timeIndex}{\agtID}$ in the time horizon $[\timeIndex,\timeIndex+2\horizon]$. To overcome this challenge, we exploit the technique of \textit{one-step ahead prediction} of the vector of the noisy output introduced in \cite{JayExtKalman}. Based on the idea introduced in \cite{JayExtKalman}, at time step $\timeIndex$, we can calculate the one-step ahead prediction of vector of the noisy output $\BoldNoisyy{\timeIndex+1}{}$ using Eq. \eqref{eq: Ypredit}.
\vspace{-3pt}
\begin{equation}\label{eq: Ypredit}
    \BoldNoisyy{\timeIndex+1}{}=\AugC{}{}(\EstSt{\timeIndex}{}+\ControlInput{\timeIndex}{})+\frac{\AugC{}{}\AugB{}{}\AugA{}{}}{2}(\ControlInput{\timeIndex}{}-\ControlInput{\timeIndex-1}{})
\end{equation}

In what follows, we explain the details of the proposed approach. As was mentioned earlier, we want to synthesize the controller inputs in a distributed manner, i.e., each agent $\agtID$ solves $\MILP$ in the control horizon of $2H$, at each time step $\timeIndex$, to synthesize its own control inputs.

Alg. \ref{Alg:MPC} illustrates the proposed receding horizon control procedure that each agent $\agtID$ uses to synthesize its own control inputs in the time length $\tau$. In Alg. \ref{Alg:MPC}, at each time step $\timeIndex$, each agent $i$ computes (1) a finite agent-level trajectory $\AgtTrajNoT{\agtID}$ with a length of $2\horizon$ that satisfies $\formula{}{\agtID}$ in the first $\horizon$ time steps and (2) computes a vector of finite estimated system-level trajectories $\GMNoT{\agtID}$ with a length of $2\horizon$ that satisfies $\formula{}{\textrm{s}}$ in the first $\horizon$ time steps while taking into consideration that the actual system-level trajectory $\ActGMNoT{}$ satisfies $\formula{}{\textrm{s}}$ with a minimum probability $\MinConf$. In other words, we enforce the satisfaction of $\formula{}{\agtID}$ and $\formula{}{\textrm{s}}$ in the first $\horizon$ time steps of $\AgtTrajNoT{\agtID}$ and $\GMNoT{\agtID}$, respectively.
\vspace{-7pt}
\begin{remark}\label{re:IndSpec}
Because each agent $\agtID$ has access to its own actual trajectory $\AgtTrajNoT{\agtID}$ at each time step $\timeIndex$, we enforce the MILP constraints related to each $\robustness[\agtID](\AgtTraj{0:\horizon-1}{\agtID},{\formula{}{{\agtID}}},\STLtimeIndex)>\MinRob$ directly using $\AgtTrajNoT{\agtID}$ for each agent $\agtID$, where $\STLtimeIndex\in[0,H-1]$.
\end{remark}

In $\MILP$, we implement the constraint $\robustness[\agtID](\AgtTraj{0:\horizon-1}{\agtID},{\formula{}{{\agtID}}},\STLtimeIndex)>\MinRob$ using the MILP technique introduced in \cite{MPCDonze}. $\ConfidenceLe{}{\agtID}((\ActGM{0:\horizon-1}{},\STLtimeIndex)\models{\formula{}{\textrm{s}}})>\MinConf$ enforces that the actual system-level trajectory $\ActGM{\STLtimeIndex}{}$ satisfies $\phi_{\textrm{s}}$ with the minimum probability $\MinConf$ for all $\STLtimeIndex\in[0,\horizon-1]$. For implementing the constraint $\ConfidenceLe{}{\agtID}((\ActGM{0:\horizon-1}{},\STLtimeIndex)\models{\formula{}{\textrm{s}}})>\MinConf$, we enforce $\robustness[\agtID](\GM{0:\horizon-1}{\agtID},{\formula{}{{\textrm{s}}}},\STLtimeIndex)>\MinRob$ and we calculate the estimated error bound $\EstErrorBound{\STLtimeIndex}{\agtID}$ using Eq. \eqref{eq:errorbound} for all $\STLtimeIndex\in[0,\horizon-1]$. Then, we encode the constraint $\ConfidenceLe{}{\agtID}((\ActGMNoT{},\timeIndex)\models{\pi})\geq{1}-\frac{\EstErrorBound{j}{}}{\MinRob}$, and after that we recursively encode the constraints listed in Lemma \ref{lem:MonRul} according to $\formula{}{\textrm{s}}$ for all $\STLtimeIndex\in[0,\horizon-1]$.


At time step $\timeIndex=0$, each agent $\agtID$ solves the optimization problem \eqref{prob:optmizeV} to obtain matrix $\TopMatrix$ (Line 1 in Alg. \ref{Alg:MPC}). At the next time step $\timeIndex=1$, each agent $\agtID$ calculates the Kalman gain matrices $\KalmanGain{\timeIndex}{}$ in the time horizon $[\timeIndex,\timeIndex+2\horizon-1]$ using Eq. \eqref{eq:KalmanGain} and \eqref{eq:ErrorCo}. Also, each agent $\agtID$ sets the $i$-th row of  $\GM{0}{\agtID}$ equal to $\EstStAgt{0}{\agtID}$ and the rest of the entries of $\GM{0}{\agtID}$ are equal to 0.

If at time step $\timeIndex=1$, agent $\agtID$ is active, then agent $\agtID$ communicates with one of its neighboring agents $\counterl\in{\neighborSet{\agtID}}$ with a uniform random probability $\frac{1}{|\AgentSet{i}|}$ to acquire $\Noisyy{1}{\counterl}$ and update the $\counterl$-th row of $\EstSt{1}{\agtID}$ (Lines 4-8 in Alg. \ref{Alg:MPC}). At Line 9 in Alg. \ref{Alg:MPC}, $P$ is a $1\times\horizon$ vector of variables that represents minimum required robustness degree at each time step in the time horizon $[0,\horizon]$. Here, we choose a robustness degree of value $\MinRob$ for all $\timeIndex\in[0,\horizon-1]$ (Line 9 in ALg. \ref{Alg:MPC}). At Line 10 in Alg. \ref{Alg:MPC}, $\mathcal{P}$ is a $1\times\horizon$ vector of variables that represents the minimum required confidence level of each agent $\agtID$ in satisfying $\formula{}{\textrm{s}}$ at each time step in the time horizon $[0,\horizon-1]$. Here, we choose a confidence level of value $\MinConf$ for all $\timeIndex\in[0,\horizon-1]$.

At Line 11 of Alg. \ref{Alg:MPC}, at each time step $t$, each agent $\agtID$ calculates a sequence of control inputs $\BoldControlInput{\timeIndex}{\agtID}=[\BoldControlInputHo{0}{\agtID}{\timeIndex},\BoldControlInputHo{1}{\agtID}{\timeIndex},...,\BoldControlInputHo{2\Horizon-1}{\agtID}{\timeIndex}]$ in the control horizon $2\horizon$ by solving $\MILP$. Here, $\BoldControlInputHo{2\Horizon-1}{\agtID}{\timeIndex}$ represents the predicted vector of control inputs calculated at the future time step $\timeIndex+2\Horizon-1$ by agent $\agtID$, and the current time step is $\timeIndex$.

\begin{algorithm}
	\small
	\DontPrintSemicolon
	
\Input{A positive large number \M\newline
$\max\limits_{i\in{\AgentSet{}}}\{\horizon(\formula{}{\textrm{s}}),\horizon(\formula{}{i})\}$ and the time length $\tau$ \newline
The minimum confidence level $\MinConf$\newline
The minimum robustness degree $\MinRob$\newline
The adjacency matrix $\GraphMat$\newline
The number of the agents $|\AgentSet{}|$\newline
The initial state covariance matrix $\StCovarianceMat{0}{}$ and the noise covariance matrix $\NoiseCo{}{}$\newline
Lipschitz constants $\mathcal{L}_1$ and $\mathcal{L}_2$, $\StNoT{\textrm{max}}$, $\zeta_{\textrm{max}}$, and $v_{\textrm{max}}$
} 
Agent $\agtID$ calculates $\TopMatrix$ by solving \eqref{prob:optmizeV}\\
\While{$\timeIndex<\tau-\horizon$ or \textrm{Diff-MILP} is feasible}{

Agent $\agtID$ calculates the Kalman gain matrices $\KalmanGain{\counterk}{}$ for all $\counterk\in\{\timeIndex,\timeIndex+1,...,\timeIndex+2\horizon-1\}$ using Eq. \eqref{eq:KalmanGain}\\
\If{agent $\agtID$ is active}
{Agent $\agtID$ communicates with agent $\counterl\in{\neighborSet{\agtID}}$ to acquire the noisy output $\Noisyy{\timeIndex}{\counterl}$\\

Agent $\agtID$ updates the $\counterl$-th row of  $\BoldNoisyy{\timeIndex}{\agtID}$\\
Agent $\agtID$ updates the $\counterl$-th row of $\EstSt{\timeIndex}{\agtID}$ using Eq. \eqref{eq:Est} with the updated $\BoldNoisyy{\timeIndex}{\agtID}$
}


$\MinRob\gets\PiCon{\counterj}{}$, $\forall{\counterj}\in\{0,...,\horizon-1\}$\\

$\MinConf\gets\Confidence{\counterj}{}$, $\forall{\counterj}\in\{0,...,\horizon-1\}$\\

Compute $\BoldControlInput{\timeIndex}{\agtID}=[\BoldControlInputHo{0}{\agtID}{\timeIndex},\BoldControlInputHo{1}{\agtID}{\timeIndex},...,\BoldControlInputHo{2\Horizon-1}{\agtID}{\timeIndex}]$ by solving $\MILP$\\
} 
 
	\caption{Distributed Differentially Private Receding Horizon Control for MTL Specifications}
    \label{Alg:MPC}

\end{algorithm}

\vspace{-15pt}

\section{Case Study}\label{sec:case}
In this section, we implement the proposed approach in a case study. We consider an MAS consisting of four agents with the set of nodes $\mathcal{C}=\{c_1,c_2,c_3,c_4\}$ and the set of edges $\mathcal{E}=\{e_{12},e_{14},e_{23},e_{34}\}$. We set $\AugA{}{}$, $\AugB{}{}$, and $\AugC{}{}$ to be $4\times{4}$ diagonal matrices where $\AugA{}{ii}=\AugB{}{ii}=\AugC{}{ii}=0.1~\text{for}~{i\in\{1,2,3,4\}}$. 

\begin{wrapfigure}{r}{0.2\textwidth}
  \centering
  \vspace{-1mm}
  \includegraphics[width=0.2\textwidth]{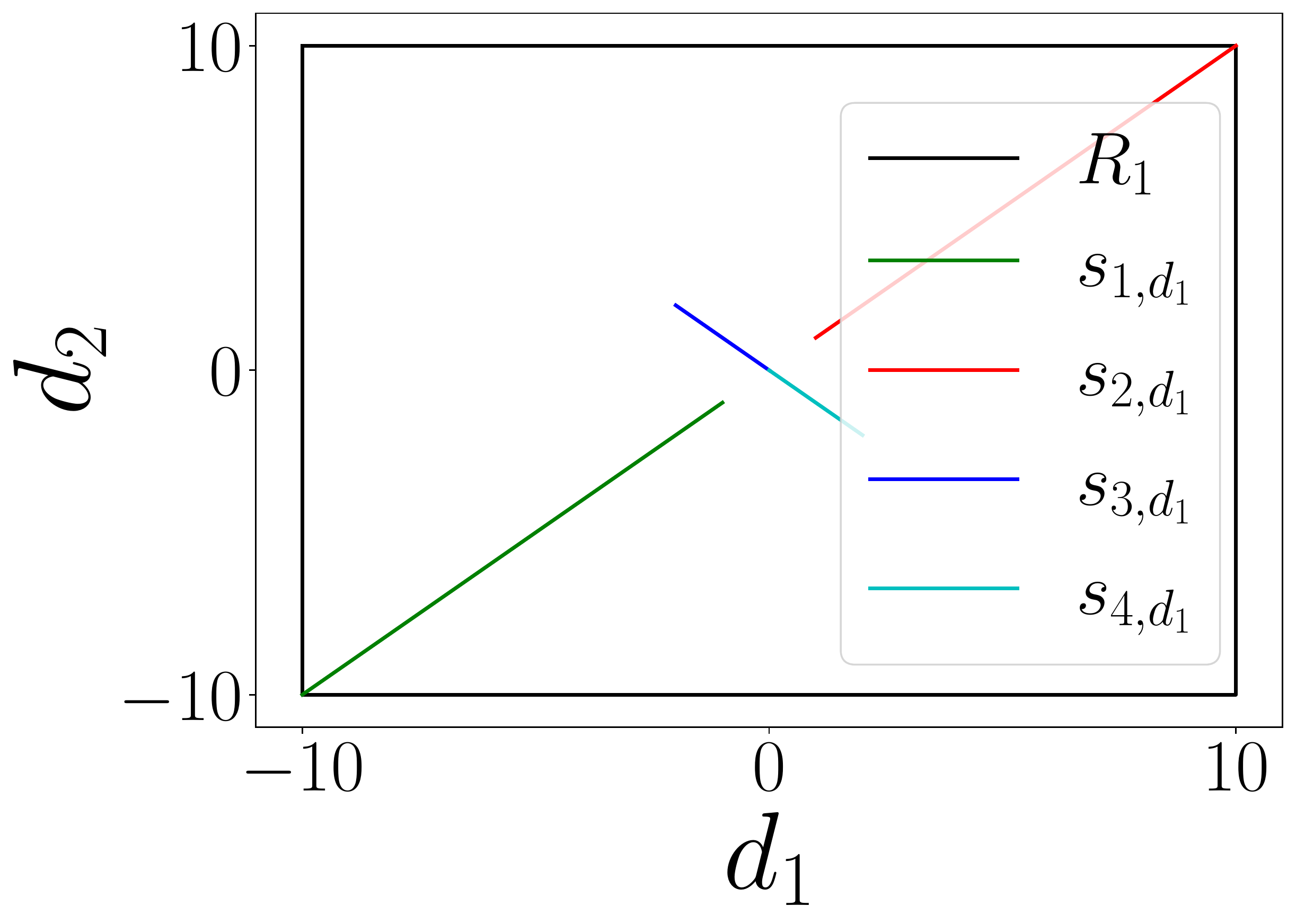}

        \caption{The illustration of the 2D planar environment with the paths that the four agents take.}
        \label{fig:2DPlan}
\end{wrapfigure}
We consider a 2-dimensional planar environment $\mathcal{S}\subseteq\mathbb{R}^2$ in which we have the following areas: (1) $R_1$ is a rectangular area centred at $(0,0)$ with the width and length equal to 20, and (2) $R_2$ is rectangular area centred at $(0,0)$ with the width and length equal to 100. Fig. \ref{fig:2DPlan} represents the 2D environment $\mathcal{S}$ containing areas $R_1$ and $R_2$ in addition to the paths of agents 1, 2, 3, and 4 in the time length of $\tau=500$s ($R_2$ is not shown in Fig. \ref{fig:2DPlan} for the better representation of the paths taken by the agents). We also denote the four quadrants of the 2D plane (starting from the positive quadrant going clockwise) by $\Tilde{Q}_1$, $\Tilde{Q}_2$, $\Tilde{Q}_3$, and $\Tilde{Q}_4$, respectively.


We specify the system-level specification as $\formula{}{\textrm{s}}:=\leventually_{[0,15]}(\ActGMNoT{}\in{R_2})$, and the agent-level specifications as $\formula{}{1}:=(\leventually_{[5,10]}(\AgtTrajNoT{1}\in{R_1}))\land(\lglobally_{[0,15]}(\AgtTrajNoT{1}\in{\Tilde{Q}_3}))$, $\formula{}{2}:=(\leventually_{[10,15]}(\AgtTrajNoT{2}\in{R_1}))\land(\lglobally_{[0,15]}(\AgtTrajNoT{2}\in{\Tilde{Q}_1}))$, $\formula{}{3}:=(\leventually_{[0,10]}(\AgtTrajNoT{3}\in{R_1}))\land(\lglobally_{[0,15]}(\AgtTrajNoT{3}\in{\Tilde{Q}_4}))$, and $\formula{}{4}:=(\leventually_{[0,10]}(\AgtTrajNoT{4}\in{R_1}))\land(\lglobally_{[0,15]}(\AgtTrajNoT{4}\in{\Tilde{Q}_2}))$. $\formula{}{\textrm{s}}$ reads as ``\textit{the centroid  of the MAS should eventually reach the area $R_2$ in the next 15 time steps}". $\formula{}{1}$ reads as ``\textit{agent 1 should eventually reach area $R_1$ in the time span of $[5,10]$ and always stay in the third quadrant in the time span $[0,15]$}". $\formula{}{2}$, $\formula{}{3}$, and $\formula{}{4}$ can be translated to natural language in a similar manner.


For synthesizing the control inputs for satisfying the given STL specifications ($\formula{}{\textrm{s}}$, $\formula{}{1}$, $\formula{}{2}$, $\formula{}{3}$, and $\formula{}{4}$), we calculate the minimum necessary lengths of the given STL specifications as $\horizon(\formula{}{\textrm{s}})=\horizon(\formula{}{1})=\horizon(\formula{}{2})=\horizon(\formula{}{3})=\horizon(\formula{}{4})=15$. We choose the control horizon to be  $2\horizon(\formula{}{1})$. Also, we set $\M=1000$, $\MinConf=0.9$, $\MinRob=0.1$, $\St{0}{1,d_1}=\St{0}{1,d_2}=-100$, $\St{0}{2,d_1}=\St{0}{1,d_2}=100$, $\St{0}{3,d_1}= -100$ and $\St{0}{3,d_2}=10$, $\St{0}{4,d_1}= 100$ and $\St{0}{4,d_2}=-10$ , and $u_{\text{min}}=-2~\text{and}~u_{\textrm{max}}=2$. Additionally, we have $\GMnonBold{0}{1,d_1}=\GMnonBold{0}{1,d_2}=-100$, $\GMnonBold{0}{2,d_1}=\GMnonBold{0}{2,d_2}=100$, $\GMnonBold{0}{3,d_1}=-100$ and $\GMnonBold{0}{3,d_2}=10$, and $\GMnonBold{0}{4,d_1}=100$ and $\GMnonBold{0}{4,d_2}=-10$. In addition, we add the Gaussian noise to the outputs $y_{\agtID}$ with the differential privacy parameters $\DiffEps{}{}\in[\log(6),\log(10)]$ and $\DiffDelta{}{}\in[0.1,0.4]$. 

Fig. \ref{fig:zeta_1} represents the obtained results for estimated system-level trajectories for agent 1, 2, 3, and 4 in the dimension $d_1$ and $d_2$, respectively. We should note that we did not include the initial values of the estimated system-level trajectories for better representation of the results in the time length of $\tau=500$s.  
As can be seen, $\formula{}{\textrm{s}}$ is satisfied with the probability of 1 by the actual system-level trajectory $\ActGMNoT{}$ which is higher than $\MinConf=0.9$. In addition, Fig. \ref{fig:zeta_1} also shows that when the time goes to infinity, the estimated system-level trajectories converge to the actual system-level trajectory. In addition, we measure the average of the absolute value of the differences between the the actual system-level trajectory and the estimated system-level trajectory at different time indices $k$ for agent $i$ at dimension $d_l$ as $\frac{1}{N}\sum\limits_{\counterk=1}^{500}|\GMnonBold{\counterk}{i,d_l}-(\eta)_{d_l}[k]|$. This value for agent 1 is approximately 12.3 at both dimensions. Similarly for agents 2, 3 and 4, at both dimensions, we have 12.6, 0.8, and 1.1, respectively.

Fig. \ref{fig: agt1} illustrates the actual agent-level trajectories of agent 1, 2, 3, and 4 in the time length of $\tau = 500s$. $\AgtTrajNoT{1}$, $\AgtTrajNoT{2}$, $\AgtTrajNoT{3}$, and $\AgtTrajNoT{4}$ converge to zero after time-step 285 due to enforcing the constraints $\lglobally_{[0,21]}(\AgtTrajNoT{1}\in{\Tilde{Q}_3})$ in $\formula{}{1}$, $\lglobally_{[0,15]}(\AgtTrajNoT{2}\in{\Tilde{Q}_1})$ in $\formula{}{2}$, $\lglobally_{[0,15]}(\AgtTrajNoT{3}\in{\Tilde{Q}_4})$ in $\formula{}{3}$, and $\lglobally_{[0,15]}(\AgtTrajNoT{3}\in{\Tilde{Q}_2})$ in $\formula{}{4}$. In addition, the results show that $\AgtTrajNoT{1}$, $\AgtTrajNoT{2}$, $\AgtTrajNoT{3}$, and $\AgtTrajNoT{4}$ satisfy $\formula{}{1}$, $\formula{}{2}$, $\formula{}{3}$, and $\formula{}{4}$ respectively.

\begin{figure}
     \centering
         \includegraphics[scale=0.15]{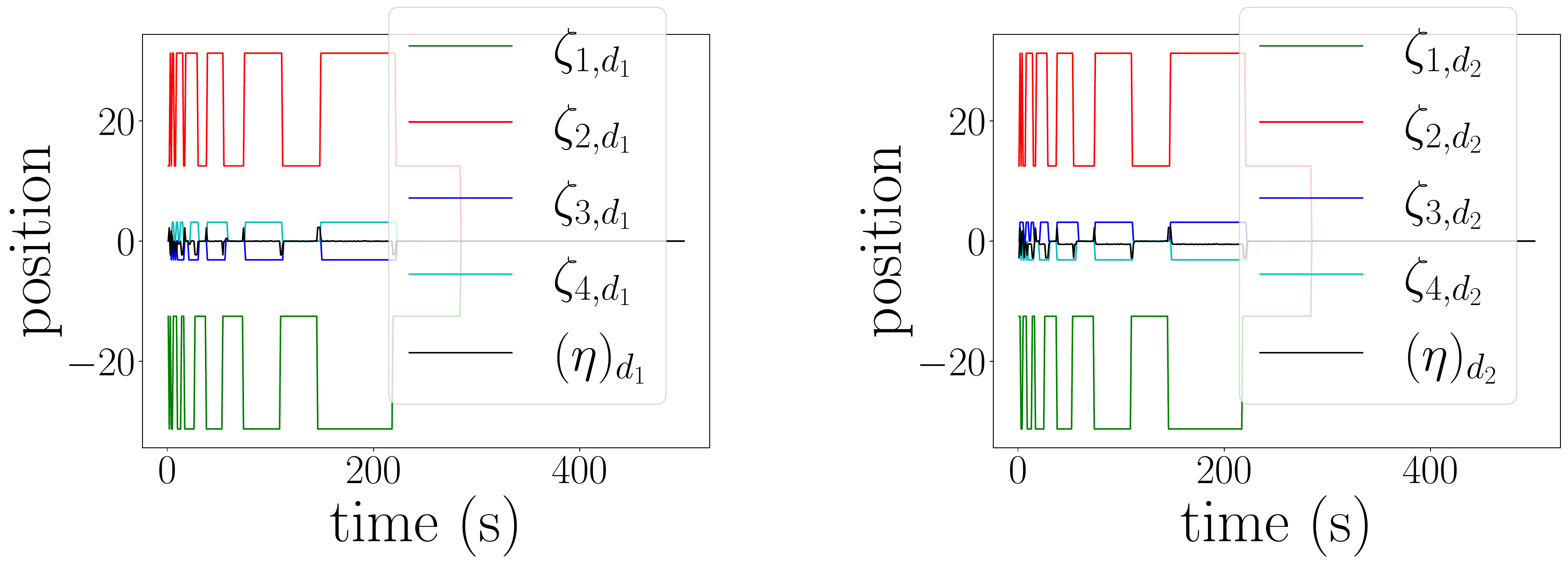}

        \caption{A comparison between the estimated system-level trajectories ($\zeta_1$, $\zeta_2$, $\zeta_3$, and $\zeta_4$) and the actual system-level trajectory ($\eta$) in the first dimension $d_1$ (left) and the second dimension $d_2$ (right). The system-level STL specification $\formula{}{\textrm{s}}$ is satisfied with the probability of 1 (higher than $\MinConf=0.9$) by the actual system-level trajectory $\ActGMNoT{}$. In addition, when the time goes to infinity, the estimated system-level trajectories converge to the actual system-level trajectory.}
        \label{fig:zeta_1}
\end{figure}

\begin{figure}

         \centering
         \includegraphics[scale=0.15]{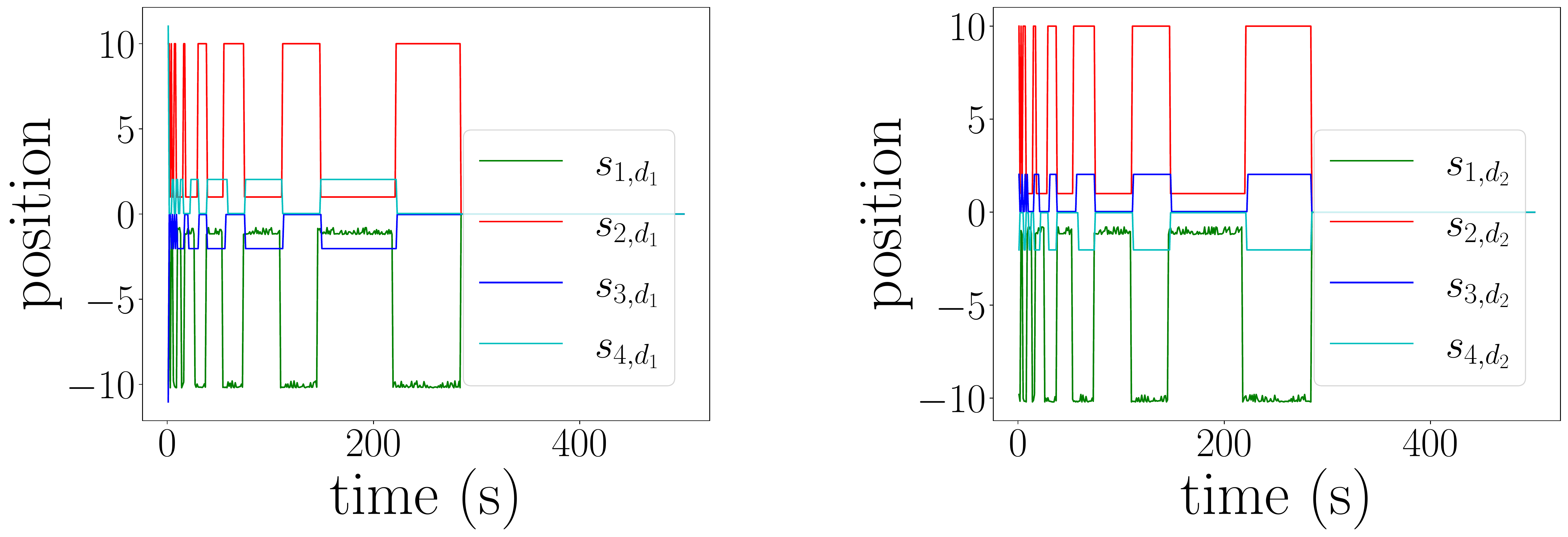}
        \caption{Obtained actual agent-level trajectories for agents 1, 2, 3, and 4 in the first dimension $d_1$ (left) and second dimension $d_2$ (right). $\AgtTrajNoT{1}$, $\AgtTrajNoT{2}$, $\AgtTrajNoT{3}$, and $\AgtTrajNoT{4}$ converge to zero due to enforcing the constraints $\lglobally_{[0,15]}(\AgtTrajNoT{1}\in{\Tilde{Q}_3})$ in $\formula{}{1}$, $\lglobally_{[0,15]}(\AgtTrajNoT{2}\in{\Tilde{Q}_1})$ in $\formula{}{2}$, $\lglobally_{[0,15]}(\AgtTrajNoT{3}\in{\Tilde{Q}_4})$ in $\formula{}{3}$, and $\lglobally_{[0,15]}(\AgtTrajNoT{4}\in{\Tilde{Q}_2})$ in $\formula{}{4}$.}
        \label{fig: agt1}
\end{figure}

\section{Conclusion}
In this paper, we proposed a distributed receding horizon control (RHC) for multi-agent systems (MAS) with MTL specifications. A potential future direction of the proposed approach is synthesizing control inputs for the situation where agents share their partial outputs instead of noisy outputs. In addition, incorporating learning-based methods in the setting where the system dynamics of the agents is can be another potential direction of the proposed approach \cite{ZheDist}.    

\titleformat{\section}{\centering\normalfont\scshape}{\appendixname~\thesection }{0em}{~}

\bibliographystyle{IEEEtran}
\bibliography{biblography}

\end{document}